\documentclass[proceedings]{stacs}
\stacsheading{2010}{441--452}{Nancy}
\usepackage{tikz}
\usepackage[utf8]{inputenc}

\newcommand{\ie}{\textit{i.e.}\ }

\newcommand{\kl}[1]{\mathcal{#1}}
\usepackage{bbm}
\newcommand{\N}{\mathbbm N}
\newcommand{\Z}{\mathbbm Z}
\newcommand{\deux}{\mathbbm2}

\newcommand{\et}{\textrm{ and }}

\usepackage{stmaryrd}
\newcommand{\co}[2]{\left\llbracket #1,#2\right\llbracket}
\newcommand{\cc}[2]{\left\llbracket #1,#2\right\rrbracket}
\newcommand{\oc}[2]{\left\rrbracket #1,#2\right\rrbracket}
\newcommand{\oo}[2]{\left\rrbracket #1,#2\right\llbracket}

\newcommand{\scc}[2]{_{\cc{#1}{#2}}}
\newcommand{\sco}[2]{_{\co{#1}{#2}}}

\newcommand{\sett}[2]{\left\{\left.#1\vphantom{#2}\right|#2\right\}}
\newcommand{\set}[3]{\sett{#1\in#2}{#3}}

\newcommand{\card}[1]{\left|#1\right|}
\newcommand{\length}[1]{\left|#1\right|}
\newcommand{\restr}[1]{_{\left|#1\right.}}
\newcommand{\lang}{\kl L}

\newcommand{\dinf}[1]{\vphantom{#1}^\infty{#1}^\infty}
\newcommand{\pb}[2]{\begin{quote}{\normalfont\textbf{Instance:}} #1.\\{\normalfont\textbf{Question:}} #2?\end{quote}}
\newcommand{\p}{\kl P}

\begin{document}

\title{Revisiting the Rice Theorem of Cellular Automata}

\author[cmm]{P. Guillon}{Pierre Guillon}
\address[cmm]{DIM - CMM, UMI CNRS 2807, Universidad de Chile\\
Av. Blanco Encalada 2120,
8370459 Santiago, Chile}
\email{pguillon@dim.uchile.cl}

\author[caen]{G. Richard}{Gaétan Richard}
\address[caen]{Greyc\\
Université de Caen \& CNRS \\
boulevard du Maréchal Juin, 14\,000 Caen, France}
\email{grichard@info.unicaen.fr}

\thanks{Thanks to the Projet Blanc ANR {\it Sycomore} and Program {\it ECOS-Sud}}


\begin{abstract}
  \noindent A cellular automaton is a parallel synchronous computing model, which consists in a juxtaposition of finite automata whose state evolves according to that of their neighbors. It induces a dynamical system on the set of configurations, \ie the infinite sequences of cell states. The limit set of the cellular automaton is the set of configurations which can be reached arbitrarily late in the evolution.
  In this paper, we prove that all properties of limit sets of cellular automata with binary-state cells are undecidable, except surjectivity. This is a refinement of the classical ``Rice Theorem'' that Kari proved on cellular automata with arbitrary state sets.
\end{abstract}

\maketitle

\section*{Introduction}

Among all results on undecidability, Rice's Theorem~\cite{rice-org} is probably one of the most important. It can be seen as stating the following: for any property on the functions computed by Turing machines, the set of corresponding machines is either trivial or undecidable. Following Church-Turing thesis, it is often thought that this result should remain true for other computational systems. It has, for instance, been extended with various restrictions to general dynamical systems~\cite{tileshifttm}, tilings~\cite{cplxtile} or, in a weaker form~\cite{tilings}.

In this paper, we shall focus on a specific model known as cellular automata, introduced by Von~Neumann~\cite{neumann}. Cellular automata are made of infinitely many cells endowed with a finite state and interacting locally and synchronously with each other. As this system does not have any way to give output, study of dynamics often uses the limit set, that consists of configurations which can appear arbitrarily late~\cite{limit1,limit2}. In this domain, Jarkko Kari has already proved an equivalent of Rice theorem~\cite{rice} for limit set. A similar, ``perpendicular'', result is also known for the trace, which consists on the evolution of only one fixed cell~\cite{trice}.

On the other hand, it is known that the property of being surjective (\textit{i.e.} having a full limit set) is decidable and not trivial for one-dimensional cellular automata~\cite{injdec}. Such a statement is not contradictory with Kari's theorem since it is not, properly speaking, a property of the limit set: a surjective CA can have the same limit set than a non-surjective one if the alphabets are distinct. Nevertheless, when fixing the alphabet, surjectivity becomes a property of the limit set. This leads to the question whether there exist other decidable properties on limit sets of cellular automata with fixed alphabet~\cite{jarkko,open}.

In this paper, we shall answer negatively to this question by extending the result of Kari when fixing the number of states, showing that all properties on limit sets other than surjectivity are either trivial or undecidable. Note that surjectivity is undecidable for higher dimensional cellular automata~\cite{revsurj}. Our proofs use borders (example of similar constructions can be found in~\cite{varouchas,trace,parsim}). The idea here is to restrict our study to nonsurjective cellular automata, since surjectivity is decidable%, which have the convenient property to decrease the complexity of the space
. The (computable) forbidden words of the image can be used as border words.

The paper is organised as follows: first we give all the necessary definitions in Section~\ref{sec:defs} and some first properties of limit sets in Section~\ref{sec:prel}. After that, we detail the core encoding of our construction in Section~\ref{sec:simul}. With all this, we state the main Rice theorem in Section~\ref{sec:Rice} before giving some concluding remarks in Section~\ref{sec:conc}.

\section{Definitions}
\label{sec:defs}

We denote the set with two elements as $\deux=\{0,1\}$.  For any alphabet $B$, we denote as $B^\Z$ the set of \emph{configurations} (all bi-infinite sequences over $B$).  The length of some word $u\in B^*$ will be noted $\length u$.  A \emph{uniform} word or configuration is one where a single letter appears, with repetitions.  For any configuration $x\in B^\Z$ or any word $x\in B^*$, $l,k \in \Z$, $x \sco lk$ denotes the finite pattern $x_{l} x_{l+1} \ldots x_{k-1}$. This notation is extended to the case where $l$ or $k$ is infinite.

A \emph{cylinder} is the subset of configurations $[u]_i = \{ x \in B^\Z \mid x \sco i{i+\length u} = u \}$ sharing the common pattern $u \in B^*$ at position $i \in \Z$. Similarly, if $E\subset B^k$ for some $k\in\N\setminus\{0\}$, then $[E]_i$ will stand for the set of configurations $x\in B^\Z$ such that $x\sco i{i+k}\in E$.

The set of configurations $B^\Z$ is a compact metric space when endowing it with the metric induced by the Cartesian product of the discrete topology on $B$. In this setting, open sets correspond to unions of cylinders.

If $b\in B$, then we note $\dinf b$ the configuration consisting in a uniform bi-infinite sequence of $b$. If $E\subset B^k$, with $k\in\N\setminus\{0\}$, then we will represent the set of bi-infinite sequences of words of $E$ as $\dinf E=\set x{B^\Z}{\exists i<k,\forall j\in\Z,x\sco{i+jk}{i+(j+1)k}\in E}$.

The following definition will be very helpful for future constructions: it allows to build borders so that some particular nonoverlapping zones of configurations can be recognized.
\begin{definition}
  Let $B$ be an alphabet and $n\in\N\setminus\{0\}$. A \emph{strongly freezing word} $u \in B^n$ is a word such that for all $i \in \co1n$, $uB^i \cap B^iu = \emptyset$. Equivalently, $[u]_0\cap[u]_i=\emptyset$.\\
A set $E \subset B^n$ is \emph{strongly freezing} if for all $i \in \co1n$, $EB^i \cap B^iE = \emptyset$.
\end{definition}
One first remark is that any word $u$ can be extended to some strongly freezing word: simply take $ub^k$, where $b\ne u_0$ and $k\in\N\setminus\{0\}$ such that $b^k$ does not appear in $u$.

The \emph{shift} $\sigma: B^\Z \to B^\Z$ is defined for all $x \in B^\Z$ and $i \in \Z$ by $\sigma(x)_i=x_{i+1}$.

A \emph{subshift} $\Sigma$ is a closed subset of $B^\Z$ which is strongly invariant by shift, \ie $\sigma(\Sigma)=\Sigma$.
Equivalently, a subshift can be defined as the set of configurations avoiding a particular set $L\subset B^*$ of finite patterns, called \emph{forbidden language}: $\set x{B^\Z}{\forall i\in\Z,\forall u\in L,x\sco i{i+\length u}\ne u}$.

If the forbidden language $L$ can be taken finite, then we say that $\Sigma$ is \emph{of finite type}; if it is empty it is the \emph{full shift}. A subshift of finite type has \emph{order} $k\in\N\setminus\{0\}$ if it admits a forbidden language $L\subset B^k$ containing only words of length $k$ --- or equivalently, of length at most $k$.

For any subshift $\Sigma$ and $-\infty\le l\le m\le+\infty$, denote $\lang\scc lm(\Sigma)=\sett{x\scc lm}{x\in\Sigma}$. Note that, when $l-m$ is finite, it only depends on this difference, justifying the definition $\lang_k(\Sigma)=\lang\sco0k(\Sigma)$ for $k\in\N$. We note $\lang(\Sigma)=\bigcup_{k\in\N}\lang_k(\Sigma)=\sett{x\scc lm}{x\in\Sigma,l,m\in\Z}$ the language of the subshift $\Sigma$.

%When $k\in\N$, $u\in A^k$ and $v\in B^k$, we will note $u\pprod v$ the product letter by letter $(u_i,v_i)_{0\le i<k}$.

\begin{definition}
  A (one-dimensional) \emph{cellular automaton} is a triplet $(B,r,f)$ where $B$ is a finite \emph{alphabet} (or state set), $r\in\N$ is the neighborhood \emph{radius} and $f: B^{2r+1} \to B$ is the \emph{local transition function}.

A cellular automaton acts on elements of $B^\Z$ (called \emph{configurations}) by synchronous and uniform application of the local transition function, inducing the \emph{global transition function} $F: B^\Z \to B^\Z$, formally defined for all $x \in B^\Z$ and $i\in\Z$ by $F(x)_i = f(x_{i-r},x_{i-r+1},\ldots,x_{i+r})$.
We will assimilate the cellular automaton with its global function.
  \end{definition}

It is easy to see that any cellular automaton commutes with the shift. In a more general way, Curtis, Hedlund and Lyndon proved that cellular automata correspond exactly to continuous self-maps of $B^\Z$ which commute whith the shift~\cite{hedlund}.

Note that a local rule $f:B^{2r+1}\to B$ can be extended in $f:B^*\to B^*$ by $f(u)=f(u\sco0{2r+1})\ldots f(u\sco{\length u-2r-1}{\length u})$.

A \emph{partial cellular automaton} is the restriction of the global function of some cellular automaton to some subshift of finite type $\Sigma$. Note that it can be defined from an alphabet $B$, a radius $r\in\N$ and a local rule $f:\lang_{2r+1}(\Sigma)\to B$.

For a cellular automaton $(B,r,f)$, a state $b \in B$ is said to be \emph{quiescent} if $f(b^{2r+1})=b$. It is said to be \emph{spreading} if $f(u)=b$ whenever the letter $b$ appears in the word $u$.

Note that if $F$ is a cellular automaton on alphabet $B$, then $F(B^\Z)$ is a subshift. In particular, either $F$ is surjective, or $F(B^\Z)$ admits (at least) a forbidden pattern.
It is easy to see that if $j\in\N\setminus\{0\}$, then $F^j$ is also a cellular automaton, and the subshift $F^j(B^\Z)$ is included in $F^{j-1}(B^\Z)$.

The evolution being parallel and synchronous, we can see that the image of any uniform configuration remains uniform. The set of uniform configuration is then a finite subsystem, with an ultimate period $p\le\card B$. In particular, $F^p$ admits some quiescent state.

\begin{definition}
  The \emph{limit set} of a cellular automaton $F$ is the set \[ \Omega_F= \bigcap_{j \in \N} F^j(B^\Z) \] of the configurations that can be reached arbitrarily late.
\end{definition}
From the remark above, the limit set of the cellular automaton $F$ always contains (at least) one uniform configuration. It is closed, and strongly invariant by $F$. More precisely, the restriction of $F$ over $\Omega_F$ is its maximal surjective subsystem. In particular, $F$ is surjective if and only if $\Omega_F=B^\Z$.

Moreover, it can be seen from the definition that $\Omega_F=\Omega_{F^k}$ : the configurations that can be reached arbitrarily late are the same.

\section{Preliminary results}
\label{sec:prel}

%Hedlund étendu aux puissances de shift (assimilées aux shifts de puissances) / block conjugacy

%Our construction is an extension of the known proof by Jarkko~Kari.
In this section, we shall recall some classical results that will be needed in the proof.% Due to space restrictions, we choose to give only the results. For details of the proofs, the reader is left to the original papers.

The ``Firing Squad'' is a problem on algorithmics over cellular automata, introduced in 1964 by Moore and Myhill in~\cite{fsquad}: the goal is to synchronize cells of arbitrarily wide zones so that they all take the same given state at the same time. It led to different solutions (see~\cite{Firing}); when dealing with infinitely many cells, we obtain that it is possible to make them get this state arbitrarily late in time, and Kari's theorem was the first extrinsic purpose to this construction; we will reuse it as it was claimed.
\begin{prop}[\cite{rice}]\label{p:fsquad}
 There exist some cellular automaton $S$ on some alphabet $B$ and some states $\kappa,\gamma\in B$, with $\kappa$ spreading, such that:
 \begin{enumerate}
 \item For any $J\in\N$, there is some configuration $z\in B^\Z$ such that $F^J(z)=\dinf\gamma$ and $\forall i\in\Z,j<J,F^j(z)_i\ne\gamma$;
 \item $\Omega_S\cap[\gamma]_0\subset\{\kappa,\gamma\}^\Z$.
 \end{enumerate}
\end{prop}

To prove the undecidability of some property, we need to reduce to it some other property which is already known to be undecidable. It is classical to reduce the nilpotency problem, which was proved undecidable in~\cite{nilpind}. This proof reduced some tiling problem to the nilpotency, but actually, the CA involved all admitted some spreading state. Hence the following stronger result can be derived directly.
\begin{thm}\label{t:nilpind}
 The problem \pb{a cellular automaton $N$ with some spreading state $\theta$}{is $N$ nilpotent} is undecidable.
\end{thm}

The restriction to cellular automata with spreading state is very convenient to allow constructions of products of cellular automata, thanks to the following result (see for instance~\cite{trice} for a simple proof).
\begin{prop}\label{p:spreadnilp}
 A cellular automaton $N$ on some alphabet $A$ with some spreading state $\theta\in A$ is nilpotent if and only if $\forall x\in A^\Z,\exists i\in\Z,j\in\N,N^j(x)_i=\theta$.
\end{prop}

\section{Binary simulation}
\label{sec:simul}

The main construction in Kari's proof is based on a simultaneous simulation of several cellular automata thanks to some complex alphabet. In order to keep a fixed alphabet, we now need to encode additionnal information into binary configurations. This can be done thanks to the fact that one of the cellular automata is assumed to be non-surjective. The non-reachable portions of configurations can be used for the complex encodings.

% {\bf TODO ....}
% Note that any subshift $\Gamma$ over alphabet $\deux^k$, with $k\in\N\setminus\{0\}$, can be seen as a subsystem $\tilde\Gamma$ of $(\deux^\Z,\sigma^k)$.
% {\bf ...}
% 
% The first step is now to show that for any 

\begin{lem}\label{simcs}
Let $C$ be an alphabet, $\Sigma$ a subshift on alphabet $\deux$ distinct from $\deux^\Z$.
Then we can build some strongly freezing language $E\subset\deux^k\setminus\lang_k(\Sigma)$, with $k\in\N\setminus\{0\}$, and some bijection  $\xi:\lang_k(\Sigma)\times C\to E$.
%  \begin{itemize}
%   \item $(\Sigma,\sigma^k)\times(C^\Z,\sigma)$ is block conjugate to $(E^\Z,\sigma)$ (or equivalently $(\dinf E,\sigma^k)$);
%   \item all ($k$-bit) letters of $E$ are forbidden patterns of $\Sigma$;
%   \item for any $i\in\co1k$, $\sigma^i(\dinf E)\cap\tilde\dinf E=\emptyset$.
%  \end{itemize}
\end{lem}

\proof 
% \todo[inline]{old version} 

% Let $u$ a forbidden pattern for $\Sigma$. Should we extend it as stated before, we can suppose that it is strongly freezing. Should we rename the letters, we can suppose that $C\subset\deux^l$, with $l\in\N\setminus\{0\}$.  As a subset of $(\deux^{\length u}\setminus\{u\})^m$, $\lang_{m\length u}(\Sigma)$ has cardinal less than $(2^{\length u}-1)^m$, and admits thus a bijection $\tilde\xi$ from $\lang_{m\length u}(\Sigma)$ onto some subset of $\deux^n$, as soon as $2^n\ge(2^{\length u}-1)^m$.  Let us take $m\ge\frac{2\length u+l}{\length u-\log(2^{\length u}-1)}$ and $n= m-2\length u-l$.  We can %assume, should we take a larger $m$, that $\Sigma$ has order $k=2m\length u$, and
% define $\xi:(z,v)\mapsto uz\sco0{m\length u}\tilde\xi(z\sco{m\length u}k)vu$ and $E=\xi(\lang_k(\Sigma)\times C)$.\\
% Now let $z\in EA^i\cap A^iE$ with $1\le i<k$. Note that $z\sco i{i+\length u}=u$.  But we also have $u=z\sco0{k}=z\sco{k-\length u}k$ and $u$ is strongly freezing, so $\length u\le i<k-2\length u$. Moreover, $z\sco{\length u}{(m+1)\length u}$ is in $\lang_{m\length u}(\Sigma)$ and therefore does not contain the pattern $u$. Hence $i>m\length u$.  Similarly, $z\sco{i+\length u}{i+(m+1)\length u}$ cannot contain the pattern $u$, so $k-\length u\notin\co{i+\length u}{i+m\length u}$, \ie $i>k-2\length u$, which gives a contradiction. 
%  \todo[inline]{new version} 

The basic idea is to use the space outside $\Sigma$ to compress the word of $\lang(\Sigma)$ and make space for the additional information $v\in C$. However, to construct it freezing, we shall compress only the second half on the word.  Let $u$ be a forbidden pattern for $\Sigma$. Should we extend it as stated before, we can suppose that it is strongly freezing. Should we rename the letters, we can suppose that $C\subset\deux^l$, with $l\in\N\setminus\{0\}$.  As a subset of $(\deux^{\length u}\setminus\{u\})^m$, $\lang_{m\length u}(\Sigma)$ has cardinal less than $(2^{\length u}-1)^m$, and admits thus a bijection $\tilde\xi$ from $\lang_{m\length u}(\Sigma)$ onto some subset of $\deux^n$ whenever $2^n\ge(2^{\length u}-1)^m$.  Let us take $m\ge\frac{2\length u+l}{\length u-\log(2^{\length u}-1)}$ and $n= m|u|-2\length u-l$.
 We now take $k=2m|u|$ and define $\xi:(z,v)\mapsto uz\sco0{m\length u}\tilde\xi(z\sco{m\length u}k)vu$ (see Fig.~\ref{fig:encod-strong}) and $E=\xi(\lang_k(\Sigma)\times C)$.

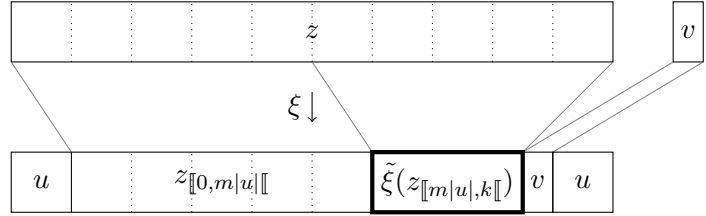
\begin{figure}[htp]
  \centering
  \begin{tikzpicture}[scale=.4]
  \draw (0,0) rectangle (20,2);
  \foreach \i in {2,4,...,18}
    \draw[dotted] (\i,0) -- (\i,2);
  \node at (10,1) {$z$};
  
  \draw (22,0) rectangle (23,2);
  \node at (22.5,1) {$v$};
  
  \draw[->] (10,-1) -- (10,-2);
  \node[anchor=east] at (10,-1.5) {$\xi$};

  \draw[help lines] (0,0) -- (2,-3);
  \draw[help lines] (10,0) -- (12,-3);
  \draw[help lines] (20,0) -- (17,-3);
  \draw[help lines] (22,0) -- (17,-3);
  \draw[help lines] (23,0) -- (18,-3);

  \begin{scope}[yshift=-5cm]
  \draw (0,0) rectangle (12,2);
  \node at (1,1) {$u$};
  \draw (2,0) -- (2,2);
  \foreach \i in {4,6,...,10}
    \draw[dotted] (\i,0) -- (\i,2);
  \node at (7,1) {$z\sco{0}{m|u|}$};
  \draw[ultra thick] (12,0) rectangle (17,2);
  \node at (14.5,1) {$\tilde\xi(z\sco{m\length u}k)$};
  \draw (17,0) rectangle (18,2);
  \node at (17.5,1) {$v$};
  \draw (18,0) rectangle (20,2);
  \node at (19,1) {$u$};
  \end{scope}

  \end{tikzpicture}
  \caption{Encoding into strongly freezing alphabet}
  \label{fig:encod-strong}
\end{figure}

 Now let $w\in E\deux^i\cap\deux^iE$ with $1\le i<k$. Note that $w\sco i{i+\length u}=u$.  But we also have $u=w\sco0{|u|}=w\sco{k-\length u}k$ and $u$ is strongly freezing, so $\length u\le i<k-2\length u$. Moreover, $w\sco{\length u}{(m+1)\length u}$ is in $\lang_{m\length u}(\Sigma)$ and therefore does not contain the pattern $u$. Hence $i>m\length u$.  Similarly, $w\sco{i+\length u}{i+(m+1)\length u}$ cannot contain the pattern $u$, so $k-\length u\notin\co{i+\length u}{i+m\length u}$, \ie $i>k-2\length u$, which gives a contradiction.  \qed 

The language $E$ will then be used as a particular alphabet, over which we can build configurations in $E^\Z$. This full shift can be more or less seen -- up to a short initial shift -- as the system $(\dinf E,\sigma^k)$, but must not be confused with the subshift $(\dinf E,\sigma)$ over $\deux$.
The key point in that construction is that the inclusion of the information of another shift can be done by a constant-space simulation: $\Sigma$ and $C^\Z$ are, in an independent way, factors of respectively $(\dinf E,\sigma)$ and of $E^\Z$ -- or, thanks to freezingness, of $(\dinf E,\sigma^k)$. This could not be done in the absence of any forbidden word $u$.

Given some partial cellular automaton $G$ on some subshift of finite type $\Sigma$, some cellular automata $N$ and $S$ on alphabets $A\ni\theta$ and $B\ni\gamma,\kappa$. Considering $C=A\times B$, we can build $E,k,\xi$ as in Lemma~\ref{simcs}. As a local rule of radius $1$, and with a little abuse of notation corresponding to commuting the products, $\xi$ can be extended to injections $\xi:\lang_{ik}(\Sigma)\times A^i\times B^i\to E^i$. Let $\delta_G$, $\delta_N$ and $\delta_S$ be the local rules of $G$, $N$ and $S$, and assume, without loss of generality, that $S$ and $N$ have same radius $r_S$ and that $G$ has radius $r_G<r_Sk$. %$\Phi$ being a block conjugacy, its reverse is based on some local rule $\xi:\lang_k(\Sigma)\times A\times B\to\lang_1(\Gamma)$. %it is based on some local rules (of radius $0$) $\phi^0:\lang_1(\Gamma)\to\lang_k(\Sigma)$, $\phi^1:\lang_1(\Gamma)\to A$, $\phi^2:\lang_1(\Gamma)\to B$ and
Define some cellular automaton $\Delta_{G,N,S}$ of radius $r=(r_S+1)k-1$ and local rule $\delta: \deux^{2r+1} \to \deux$ defined as:

\[ \delta(y) = \left| 
  \begin{array}{ll}
    \delta_G(y\scc{r-r_G}{r+r_G}) & \textrm{ if } y\in\lang_{(2r+1)}(\Sigma) \hfill \text{ (1)}\\ \\
    z_i & \textrm{ if } \left\{
      \begin{array}{l}
        y\in\deux^iE^{r_S}\xi(z,v,\gamma)E^{r_S}\deux^{k-1-i} \\
        0\le i<k,z\in\lang_k(\Sigma),v\in A \\
      \end{array}\right. \hfill \text{(2)}\\ \\
    \xi(z,\delta_N(v),\delta_S(w))_i & \textrm{ if } \left\{
      \begin{array}{l}
        y\in\deux^i\xi(z,v,w)\deux^{k-1-i} \\
        0\le i<k,z\in\lang_{(2r_S+1)k}(\Sigma) \\
        v\in A^{r_S}\times A\setminus\{\theta\}\times A^{r_S} \\
        w\in B^{r_S}\times B\setminus\{\gamma,\kappa\}\times B^{r_S} \\
      \end{array}\right. \hfill \text{(3)}\\ \\
    0 & \textrm{ otherwise} \hfill \text{(4)}\\
  \end{array}\right.
\]

% \[\appl\delta{\deux^{2r+1}}\deux u{\soit{
% \delta_G(u\scc{r-r_G}{r+r_G})\si u\in\Sigma&\text{ (1)};\\
% z_i\si\both{u\in\deux^iE^{r_S}\xi(z,v,\gamma)E^{r_S}\deux^{k-1-i}\\0\le i<k,z\in\lang_k(\Sigma),v\in A}&\text{ (2)};\\
% \xi(z,\delta_N(v),\delta_S(w))_i\si\both{u\in\deux^i\xi(z,v,w)\deux^{k-1-i}\\0\le i<k,z\in\lang_{(2r_S+1)k}(\Sigma)\\v\in A^{r_S}\times A\setminus\{\theta\}\times A^{r_S}\\w\in B^{r_S}\times B\setminus\{\gamma,\kappa\}\times B^{r_S}}&\text{ (3)};\\
% 0\sinon&\text{ (4)}.}}\] 
% % WTF 0 !!!!!

This rule is well-defined since the freezingness of $E$ imposes the unicity of $i$ in the cases (2) and (3). Intuitively, the constructed automaton behaves as $G$ on $\Sigma$ (1) and uses the freezing alphabet to simulate both automata $N$ and $S$ while keeping ``compressed'' an element of $\Sigma$ (3). This element is uncompressed when automaton $S$ reaches state $\gamma$ (2). When $N$ reaches state $\theta$, $S$ reaches state $\kappa$ or when the encoding is invalid, the local transition goes to $0$ (4).

Through the end of the section, the cellular automaton $S$ will be a Firing Squad solution as built in Proposition~\ref{p:fsquad}. This will allow to make any configuration of $\Sigma$ appear arbitrarily late during the evolution, since before the synchronization of $S$, the configuration of $\Sigma$ will not be altered.

We will also assume that $\theta$ is a spreading state for the cellular automaton $N$. Intuitively, we wonder if $N$ is nilpotent, and show that we can get an answer if we assume that some property over $\Delta_{G,N,S}$ is decidable.

Finally, we assume that the domain $\Sigma$ of $G$ is the subshift of finite type avoiding a single forbidden pattern $u$ (of length less than $k$) such that $u_0\ne0\ne u_{\length u-1}$. This last property allows that $0^*\lang(\Sigma)\subset\lang(\Sigma)$ and $\lang(\Sigma)0^*\subset\lang(\Sigma)$.

The following lemma shows that the non-encoding patterns will give words in $\Sigma$ after one evolution step.
\begin{lemma}\label{l:preinv0}
Let $x$ be such that $\Delta_{G,N,S}(x)\in[u]_0$. Then there is some $i\in\oc{-k}0$ such that %$\Delta_{G,N,S}^{-1}(x)\in[E]_i$ and
$x\in[E^{2r_S+1}]_{i-r_Sk}$.
\end{lemma}
\proof\hfill\begin{itemize}
% \item If $x\in[E^{r_S}\xi(z,v,\gamma)E^{r_S}]_{i-r_Sk}$ for some $i\in\cc{\length u-k}0,z\in\lang_k(\Sigma),v\in A$, then (2) is applied in all cells of $\co i{i+k}$, \ie $\Delta_{G,N,S}(x)\in[\lang_k(\Sigma)]_i$, hence $\Delta_{G,N,S}(x)\notin[u]_0$.
% \item If $x\in[E^{2r_S+1}]_{q-r_Sk}$ for some $q\in\cc{\length u-k}0$ and we are not in the previous case, then (3) is applied  in all cells of $\co q{q+k}$, \ie $\Delta_{G,N,S}(x)\in[E]_q$.
\item If case (4) of the rule is applied to $x$ in cell $0$, then $\Delta_{G,N,S}(x)_0=0\ne u_0$, hence $\Delta_{G,N,S}(x)\notin[u]_0$.
\item If case (1) is applied to $x$ in all cells of $\co0k$, then $\Delta_{G,N,S}(x)\sco0k\in[\lang_k(\Sigma)]_0$, hence $\Delta_{G,N,S}(x)\notin[u]_0$.
\item If $x$ applies case (1) in cell $0$, but there exists some $i\in\oo0k$ (say minimal) which applies some other rule. This means that the neighborhood $x\scc{i-1-r}{i-1+r}$ is in $\lang_{2r+1}(\Sigma)$ whilst the neighborhood $x\scc{i-r}{i+r}$ is not, \ie $x\in[u]_{i+r-\length u}$. As a result, all cells of $\co ik$ (and many more) will see a non-homogeneous neighborhood and apply (4). $x\sco0k\in\lang_i(\Sigma)0^{k-i}\subset\lang_k(\Sigma)$, hence $\Delta_{G,N,S}(x)\notin[u]_0$.
\item If $x$ applies either (2) or (3) in cell $0$, then we get the result.
\qed\end{itemize}

The following lemma completes the previous one: not only cannot $u$ appear from scratch, but no encoding pattern can appear from a non-encoding pattern.
\begin{lemma}\label{l:preinv}
Let $x$ be such that $\Delta_{G,N,S}(x)\in[E]_0$. Then $x\in[E^{2r_S+1}]_{-r_Sk}$.
\end{lemma}
\proof\hfill
\begin{itemize}
\item If case (3) of the rule is applied in some cell $j\in\co0k$, then $\Delta_{G,N,S}(x)\in[E]_{j-i}$ for some $j\in\co0k$. $E$ being freezing, we get $\Delta_{G,N,S}(x)\in[E]_0$.
\item Since all words of $E$ contain some occurrence of $u$, Lemma~\ref{l:preinv0} gives some $i\in\oc{-k}{k-\length u}$ such that $x\sco{i-r_Sk}{i+(r_S+1)k}=\xi(z,v,w)$, for some $z\in\lang_{(2r_S+1)k}(\Sigma),v\in A^{2r_S+1},w\in B^{2r_S+1}$. Assume that $i\ge0$ -- the symmetric case is similar. If the previous point does not occur, then case (2) is applied to cells of $\co i{i+k}$, and either case (2) or case (3) to cells of $\co{i-k}i$. Since $0^k\lang_k(\Sigma)\subset\lang_{2k}(\Sigma)$, both cases imply that $\Delta_{G,N,S}(x)\sco{i-k}{i+k}\in\lang_{2k}(\Sigma)$. This contradicts the assumption that $\Delta_{G,N,S}(x)\sco0k\in E$.
\qed\end{itemize}

To study more in details the limit set of the constructed automaton, let us first 
consider the set \[\Lambda=\bigcup_{\begin{subarray}c0\le i<k\\-\infty\le l<m\le+\infty\end{subarray}}\set x{\deux^\Z}{x\sco{i+lk}{i+mk}\in E^{m-l}\et\forall j\notin\co{i+lk}{i+mk},x_j=0}\] of configurations or pieces of configurations of $\dinf E$ surrounded by $0$. Note that this set does not depend on $G$ -- only on the subshift $\Sigma$.
These partially encoding configurations correspond exactly to those of the limit set of our cellular automaton which are not in $\Sigma$, as proved below.

\begin{lem}\label{l:fire}
 Let $x\in\Omega_{\Delta_{G,N,S}}$ and $i,j\in\Z$ such that $i<j$ and $x\sco i{i+k}=\xi(z_i,(a_i,\gamma)),x\sco j{j+k}=\xi(z_j,(a_j,b_j))\in E$.
Then $b_j\in\{\gamma,\kappa\}$.
\end{lem}
\begin{proof}
Assume on the contrary that $b_j\notin\{\gamma,\kappa\}$.
Let $(x^t)_{t\in\Z}$ be a biorbit of $x=x^0$, \ie a bisequence of configurations such that $\forall t\in\Z,\Delta_{G,N,S}(x^t)=x^{t+1}$. By an easy recurrence and Lemma~\ref{l:preinv}, we can see that for any $t\in\N$, $x^{-t}\scc{i-r_Stk}{i+(r_St+1)k}$ can be written $\xi(z^{-t}_{i-r_St},(a^{-t}_{i-r_St},b^{-t}_{i-r_St}))\ldots\xi(z^{-t}_{i+r_St},(a_{i+r_St},b_{i+r_St}))\in E^{(2r_St+1)k}$ and $\delta_S^t(b^{-t}_{i-r_St}\ldots b^{-t}_{i+r_St})=b_i$; similarly, $x^{-t}\scc{j-r_St}{j+r_St}$ can be written $\xi(z^{-t}_{j-r_St},(a^{-t}_{j-r_St},b^{-t}_{j-r_St}))\ldots\xi(z^{-t}_{j+r_St},(a_{j+r_St},b_{j+r_St}))\in(A\times B)^{(2r_Skt+1)k}$, and $\delta_S^t(b^{-t}_{j-r_St}\ldots b^{-t}_{j+r_St})=b_j$.
$E$ being strongly freezing, we can see that $j-i=kl$ for some $l\in\N$, and for any $t>\frac{l-1}{2r_S}$, $x^{-t}\scc{i-r_Skt}{j+r_Skt}$ is in $E^{l-1+2rt}$ and the image $\delta_S^t(b^{-t}_{i-r_St}\ldots b^{-t}_{j+r_St})$ contains $b_i$ and $b_j$. In other words, the cylinder $[b_iB^{j-i-1}b_j]_i$ intersects any of the $S^t(B^\Z)$, and by compactness intersects $\Omega_S$, which contradicts Proposition~\ref{p:fsquad}.
\end{proof}

\begin{lem}\label{l:ssgamma}
$\Omega_{\Delta_{G,N,S}}\subset\Sigma\cup\Lambda$.
\end{lem}
\proof From Lemma~\ref{l:preinv0}, the image of the subshift which avoids all patterns of $E$ is included in $\Sigma$, which itself is invariant. By shift-invariance, it is thus sufficient to prove that $\Omega_{\Delta_{G,N,S}}\cap[E]_0\subset\Lambda$.
One can remark that the patterns of $\set v{E\deux^k\deux^*}{v\sco k{2k}\notin E\et v\sco k{2k}\ne0^k}$ are forbidden in the image $\Delta_{G,N,S}(\deux^\Z)$. Indeed, if you apply case (3) of the rule in the central cell and (1) in another cell, then between these two cells there will be at least a range of $k$ cells seeing a non-homogeneous neighborhood and applying case (4). Now if you apply case (3) in the central cell and (2) in another cell, this means that you had a configuration which involved simultaneously a state of $A\times(B\setminus\{\gamma,\kappa\}$ and a state of $A\times\{gamma\}$, which contradicts Lemma~\ref{l:fire}.
By induction on $n\ge1$, we can prove that the patterns of $\set v{E\deux^{nk}\deux^*}{v\sco k{2k}\notin E\et v\sco k{2k+n-1}\ne0^{k+n-1}}$ are forbidden in $\Delta_{G,N,S}^n(\deux^\Z)$, since at least the $r_S$ extremal encoding patterns of $E$ disappear at each step, whereas the non-zero patterns of $\Sigma$ can spread only by $r_G<r_Sk$ cells every step. In the limit, we obtain that all configurations of $\Omega_{\Delta_{G,N,S}}$ containing a pattern of $E$ are in $\Lambda$.  \qed

In the case where $N$ is nilpotent, we can see that the second part of the limit set is empty, and therefore we obtain the limit set of the original cellular automaton $G$.
\begin{lem}\label{l:prodnilp}
 If $N$ is nilpotent, then $\Omega_{\Delta_{G,N,S}}=\Omega_G$.
\end{lem}
\proof
From Lemma~\ref{l:ssgamma} and the fact that $(\Delta_{G,N,S})\restr\Sigma=G\restr\Sigma$, it is sufficient to prove the emptyness of $\Omega_{\Delta_{G,N,S}}\cap\Lambda$.
Let $x\in\Omega_{\Delta_{G,N,S}}\cap[E]_0$ and $J\in\N$. There exists $y^J\in\deux^\Z$ such that $\Delta_{G,N,S}^J(y^J)=x$. Applying inductively Lemma~\ref{l:preinv}, we obtain that $y^J\sco{-Jr_Sk}{(Jr_S+1)k}\in E^{2Jr_S+1}$. By compactness, there is some configuration $y$ such that for any $i\in\Z$ and any $j\in\N$, $F^j(y)\sco{ik}{(i+1)k}\in E$. Clearly, in the successive evolution step from $y$, case (3) of the local rule is always applied, which implies that there is a configuration in $A^\Z$ in the evolution of which $\theta$ never appears, hence contradicting the nilpotency of $N$.
\qed

When $N$ is not nilpotent, we can see that there is a way to let the Firing Squad inject any configurations of $\Sigma$ at any time, hiding the action of $G$: the limit set does not depend on $G$.
\begin{lem}\label{l:sigomeg}
 If $N$ is not nilpotent, then $\Sigma\subset\Omega_{\Delta_{G,N,S}}$.
\end{lem}
\proof
Let $x\in\Sigma$. From Proposition~\ref{p:spreadnilp}, there is some configuration $y\in A^\Z$ such that for any $i\in\Z$ and any $j\in\N$, $N^j(x)_i\ne\theta$. Let $J\in\N\setminus\{0\}$. From Proposition~\ref{p:fsquad}, there is some configuration $z\in B^\Z$ such that $S^{J-1}(z)=\dinf\gamma$ and for any $j<J-1$ and any $i\in\Z$, $S^j(z)_i\notin\{\gamma,\kappa\}$ (since $\kappa$ is spreading).
Consider now the configuration $\tilde x$ defined by $\forall i\in\Z,\tilde x\sco{ik}{(i+1)k}=\xi(x\sco{ik}{(i+1)k},y_i,z_i)$.
By a quick induction on $j<J$, we can see that for any cell $i\in\Z$, only case (3) of the local rule is used, and $\Delta_{G,N,S}^j(\tilde x)\sco{ik}{(i+1)k}=\xi(x\sco{ik}{(i+1)k},N^j(y)_i,S^j(z)_i)$. At time $J$, since $S^{J-1}(y)=\dinf\gamma$, the second part of the rule is applied and $\Delta_{G,N,S}^J(\tilde x)\sco{ik}{(i+1)k}=x\sco{ik}{(i+1)k}$. As a result, $x\in\bigcap_{J\in\N\setminus\{0\}}\Delta_{G,N,S}^J(\deux^\Z)$.
\qed

\section{Rice Theorem}
\label{sec:Rice}

The construction of the previous section allows us to separate the cases whether $N$ is nilpotent in the same time as we separate properties of the limit set.
\begin{lem}\label{l:mmquiesc}
 For any nontrivial property $\p$ over the limit sets of nonsurjective cellular automata on $\deux$, there exist two cellular automata $G_0,G_1$ on alphabet $\deux$ sharing the same quiescent state $q\in\deux$, and such that $\Omega_{G_0}\in\p$, $\Omega_{G_1}\notin\p$.
\end{lem}
\proof
 Take any nonsurjective cellular automaton $M$ on $\deux$ which has both $0$ and $1$ quiescent (such as a minimum cellular automaton). If its limit set satisfies $\p$, then take some nonsurjective cellular automaton $G$ on $\deux$ whose limit set does not satisfy $\p$. Then $G^2$ has the same limit set and some quiescent state $q\in\deux$. %Let $u^1,u^2$ be two words forbidden in the (nonfull) shifts $M(\deux^\Z)$ and $G^2(\deux^\Z)$. Simply take $u=q'u^1u^2q'$, with $q'\ne q$, and $G_0$ and $G_1$ the respective restrictions of $M$ and $G^2$ to the subshift of forbidden pattern $u$. The restriction clearly does not influence the limit set.
If the limit set of $M$ does not satisfy the property $\p$, then we can do the same with some nonsurjective cellular automaton $G$ whose limit set does satisfy $\p$.
\qed

\begin{lem}\label{l:sepa}
Let $G_0,G_1$ be two nonsurjective cellular automata on alphabet $\deux$ sharing the same quiescent state $q\in\deux$, and $N$ a cellular automaton with a spreading state $\theta$.
Then we can build two cellular automata $F_i,i\in\{0,1\}$ such that $\Omega_{F_i}=\Omega_{G_i}$, $i\in\{0,1\}$, if $N$ is nilpotent; $\Omega_{F_0}=\Omega_{F_1}$ otherwise.
\end{lem}
\proof
Should we invert $0$ and $1$ in the construction, we can assume that $q=0$.
Let $u^i$ be a forbidden pattern of the (non-full) shift $G_i(\deux^\Z)$, and consider the word $u=1u^0u^11$. The restrictions $\tilde G_i$ of $G_i$ on the subshift $\Sigma$ forbidding $\{u\}$ are partial cellular automata with the same limit sets than the respective $G_i$.
 Define $F_i=\Delta_{\tilde G_i,N,S}$ with $S$ being as obtained in Proposition~\ref{p:fsquad}. Note that, except in the first case of the locale rule, the definitions of these two cellular automata are equivalent since they are based on the same subshift. If $N$ is not nilpotent, then by Lemmas~\ref{l:ssgamma} and \ref{l:sigomeg}, $\Omega_{\Delta_{\tilde G_i,N,S}}=\Sigma\cup\Omega_{(\Delta_{\tilde G_i,N,S})\restr\Lambda}$. It can be noted that the restrictions of $\Delta_{\tilde G_0,N,S}$ and $\Delta_{\tilde G_1,N,S}$ on $\Lambda$ are equal. Hence $\Omega_{\Delta_{\tilde G_0,N,S}}=\Omega_{\Delta_{\tilde G_1,N,S}}$.
Now if $N$ is nilpotent, Lemma~\ref{l:prodnilp} gives the statement.
\qed

Here is now the main result.
\begin{thm}
Let $\p$ be a property satisfied by the limit set of at least one nonsurjective cellular automaton on $\deux$, but not all. Then the problem \pb{a cellular automaton $F$ on $\deux$}{$\Omega_F\in\p$} is undecidable.
\end{thm}
\proof
 Assume such a property $\p$ is decidable. Let $G_i$, $i\in\{0,1\}$ be as in Lemma~\ref{l:mmquiesc}. Let us show a procedure to decide whether a given cellular automaton $N$ on alphabet $A$ with spreading state $\theta$ is nilpotent or not, which will contradict Theorem~\ref{t:nilpind}.
We build the two cellular automata $F_i$ as in Lemma~\ref{l:sepa} and we algorithmically check whether their limit sets satisfy property $\p$.
If $\Omega_{F_0}\in\p$ and $\Omega_{F_1}\notin\p$, then $N$ is nilpotent (otherwise the two limit sets would be equal). Otherwise, we know that one $\Omega_{F_i}$ is not equal to $\Omega_{G_i}$, so $N$ is not nilpotent.
\qed
From the decidability of the surjectivity problem, established in~\cite{injdec}, we can rephrase the previous theorem as follows: surjectivity is the only nontrivial property of the limit sets of cellular automata on alphabet $\deux$ to be decidable. Of course, this can be translated to any other fixed alphabet (of at least two letters).

\section{Perspectives}
\label{sec:conc}

This result is a very complete one, since it states that nothing can be said algorithmically with respect to how the long-time configurations look like. In spite of this, concrete examples of properties concerned are not so numerous, except nilpotency or apparition of a given state or pattern.

This is due to the fact that it does not include any dynamical idea. Various results have been obtained in this direction, about some properties of the restriction of the cellular automaton to the limit set~\cite{limind}, the properties of the sequences of states taken by a particular cell~\cite{trice}, or the regularity of the languages obtained this way~\cite{regind}.

Among the properties that are not known to be concerned by our result, an important open problem consists in asking whether stability, which corresponds to the fact that the limit set is reached whithin a finite number of states (and the undecidability of which is not very hard to establish anyway), is a property of the limit sets or not. This issue is linked to the understanding of the different types of limit sets we can get with cellular automata, treated in particular in~\cite{soflim}.

Note that space-time diagrams of cellular automata, which represent the superposition of successive configurations in its application, are two-dimensional subshifts of finite type, \ie drawings defined by some local constraints. Hence our result directly implies some kind of Rice theorem on subshift projections (multidimensional subshifts can be defined similarly).
\begin{cor}
 Let $\p$ be a property satisfied by the limit set of at least one non-surjective cellular automaton on $\deux$, but not all, and $\pi:\deux^{\Z^2}\to\deux^{\Z}$ defined by $\pi((x_{ij})_{i,j\in\Z})=(x_{0j})_{j\in\Z}$. Then the problem \pb{a subshift of finite type $\Sigma\subset\deux^{\Z^2}$}{$\pi(\Sigma)\in\p$} is undecidable.
\end{cor}
The previous collary can obviously be generalized to any projection defined similarly from some $k$-dimensional tiling to some $q$-dimensional tiling, where $0<q<k$. This does not include the property of being the full shift, but the undecidability of this property was already a consequence of the ``perpendicular'' Rice theorem in~\cite{trice}.

One may also wonder what happens for higher-dimensionnal cellular automata. In this case, our construction seems to extend well. Moreover, surjectivity is also undecidable which makes any non-trivial property undecidable.

Moreover, one can ask the same question on another characteristic set of cellular automata: the ultimate set, containing all the adhering values of orbits, studied for instance in~\cite{nilpeng}. One can notice that, when shifting enough a cellular automaton, the limit set is unchanged but the ultimate set becomes equal to the limit set. Hence, any nontrivial property of limit sets of cellular automata, except being a full shift, is an undecidable property of the ultimate set. We can wonder if it is the case for other nontrivial properties of ultimate sets.

More generally, the Firing Squad can be seen as a very powerful tool to touch the limit set. Our binary simulation can help hide its evolution within any alphabet. This could allow other complex constructions desolidarizing the simulation of a cellular automaton and the structure of its limit set. For instance, could we build an intrinsically universal cellular automaton (\ie that can simulate any other cellular automaton) whose limit set is any given subshift of finite type?

% To obtain such a result, a better understanding of the dynamics of the firing squad solution could be interesting. For instance, does there exist a cellular automaton with some spreading state $0$ such that any configuration of the limit set contains a past that does not involve state $0$? JE SUIS PAS SÛR DE CETTE QUESTION. 
%% il me semble que le firing squad de Kari fait exactement ça, dans le double doute, je mets en commentaire ...

\bibliographystyle{splncs}
\bibliography{ult}

\begin{thebibliography}{{\v{C}ul\'{\i}k II}PY89}

\bibitem[AP72]{injdec}
S.~Amoroso and Y.~N. Patt.
\newblock Decision procedures for surjectivity and injectivity of parallel maps
  for tessellation structures.
\newblock {\em Journal of Computer \& System Sciences}, 6:448--464, 1972.

\bibitem[CD04]{tilings}
Julien Cervelle and Bruno Durand.
\newblock Tilings: recursivity and regularity.
\newblock {\em Theoretical Computer Science}, 310(1--3):469--477, March 2004.

\bibitem[CFG07]{trace}
Julien Cervelle, Enrico Formenti, and Pierre Guillon.
\newblock Sofic trace of a cellular automaton.
\newblock In S.~Barry Cooper, Benedikt Löwe, and Andrea Sorbi, editors, {\em
  Computation and Logic in the Real World, $3^\text{rd}$ Conference on
  Computability in Europe (CiE07)}, volume 4497 of {\em Lecture Notes in
  Computer Science}, pages 152--161, Siena, Italy, June 2007. Springer-Verlag.

\bibitem[CG07]{trice}
Julien Cervelle and Pierre Guillon.
\newblock Towards a {R}ice theorem on traces of cellular automata.
\newblock In Ludek Kučera and Antonín Kučera, editors, {\em $32^\text{nd}$
  International Symposium on the Mathematical Foundations of Computer Science},
  volume 4708 of {\em Lecture Notes in Computer Science}, pages 310--319,
  Český Krumlov, Czech Republic, August 2007. Springer-Verlag.

\bibitem[{\v{C}ul\'{\i}k II}PY89]{limit2}
Karel {\v{C}ul\'{\i}k II}, Jan~K. Pachl, and Sheng Yu.
\newblock On the limit sets of cellular automata.
\newblock {\em SIAM Journal on Computing}, 18(4):831 -- 842, 1989.

\bibitem[DB04]{tileshifttm}
Jean-Charles Delvenne and Vincent Blondel.
\newblock Quasi-periodic configurations and undecidable dynamics for tilings,
  infinite words and {T}uring machines.
\newblock {\em Theoretical Computer Science}, 319:127--143, 2004.

\bibitem[DFM00]{open}
Marianne Delorme, Enrico Formenti, and Jacques Mazoyer.
\newblock Open problems on cellular automata.
\newblock Research Report 2000-25, École Normale Supérieure de Lyon, July
  2000.

\bibitem[DFV03]{varouchas}
Bruno Durand, Enrico Formenti, and Georges Varouchas.
\newblock On undecidability of equicontinuity classification for cellular
  automata.
\newblock In Michel Morvan and \'Eric Rémila, editors, {\em Discrete Models
  for Complex Systems (DMCS'03)}, volume~AB of {\em DMTCS Proc.}, pages
  117--128. Discrete Mathematics and Theoretical Computer Science, June 2003.

\bibitem[dL06]{regind}
Pietro di~Lena.
\newblock Decidable properties for regular cellular automata.
\newblock In Gonzalo Navarro, Leopoldo~E. Bertossi, and Yoshiharu Kohayakawa,
  editors, {\em $4^\text{th}$ IFIP International Conference on Theoretical
  Computer Science (TCS'06)}, volume 209 of {\em International Federation for
  Information Processing}, pages 185--196, Santiago, Chile, August 2006.
  Springer.

\bibitem[dLM09]{limind}
Pietro di~Lena and Luciano Margara.
\newblock Undecidable properties of limit set dynamics of cellular automata.
\newblock In Susanne Albers and Jean-Yves Marion, editors, {\em
  $26^{\text{th}}$ International Symposium on Theoretical Aspects of Computer
  Science (STACS'09)}, pages 337--348, Freiburg, Germany, February 2009. IBFI
  Schloss Dagstuhl.

\bibitem[GR08]{nilpeng}
Pierre Guillon and Gaétan Richard.
\newblock Nilpotency and limit sets of cellular automata.
\newblock In Edward Ochmański and Jerzy Tyszkiewicz, editors, {\em
  $33^\text{rd}$ International Symposium on the Mathematical Foundations of
  Computer Science (MFCS'08)}, volume 5162 of {\em Lecture Notes in Computer
  Science}, pages 375--386, Toruń, Poland, August 2008. Springer-Verlag.

\bibitem[Hed69]{hedlund}
Gustav~A. Hedlund.
\newblock Endomorphisms and automorphisms of the shift dynamical system.
\newblock {\em Mathematical Systems Theory}, 3:320 -- 375, 1969.

\bibitem[Hur87]{limit1}
Lyman~P. Hurd.
\newblock Formal language characterization of cellular automaton limit sets.
\newblock {\em Complex systems}, 1:69 -- 80, 1987.

\bibitem[Kar92]{nilpind}
Jarkko Kari.
\newblock The nilpotency problem of one-dimensional cellular automata.
\newblock {\em SIAM Journal on Computing}, 21(3):571--586, 1992.

\bibitem[Kar94a]{revsurj}
Jarkko Kari.
\newblock Reversibility and surjectivity problems of cellular automata.
\newblock {\em Journal of Computer and System Sciences}, 48(1):149--182, 1994.

\bibitem[Kar94b]{rice}
Jarkko Kari.
\newblock {R}ice's theorem for the limit sets of cellular automata.
\newblock {\em Theoretical Computer Science}, 127(2):229--254, 1994.

\bibitem[Kar05]{jarkko}
Jarkko Kari.
\newblock Theory of cellular automata: A survey.
\newblock {\em Theoretical Computer Science}, 334:3--33, 2005.

\bibitem[LW08]{cplxtile}
Grégory Lafitte and Michael Weiss.
\newblock Computability of tilings.
\newblock In Giorgio Ausiello, Juhani Karhumäki, Giancarlo Mauri, and Luke
  Ong, editors, {\em $5^\text{th}$ IFIP International Conference on Theoretical
  Computer Science (TCS'08)}, volume 273 of {\em International Federation for
  Information Processing}, pages 187--201, Milano, Italy, September 2008.
  Springer, Boston.

\bibitem[Maa95]{soflim}
Alejandro Maass.
\newblock On the sofic limit set of cellular automata.
\newblock {\em Ergodic Theory \& Dynamical Systems}, 15:663--684, 1995.

\bibitem[Maz96]{Firing}
Jacques Mazoyer.
\newblock On optimal solutions to the firing squad synchronization problem.
\newblock {\em Theoretical Computer Science}, 168(2):367 -- 404, 1996.

\bibitem[Moo64]{fsquad}
Edward~F. Moore.
\newblock The firing squad synchronisation problem.
\newblock In Addison-Wesley, editor, {\em Sequential machines, Selected
  papers}, pages 213--214, 1964.

\bibitem[Pou08]{parsim}
Victor Poupet.
\newblock Translating partitioned cellular automata into classical cellular
  automata.
\newblock In Bruno Durand, editor, {\em Journées Automates Cellulaires,
  $1^\text{st}$ Symposium on Cellular Automata (JAC'08)}, pages 130--140,
  Uzès, France, April 2008. MCCME Publishing House, Moscow.

\bibitem[Ric53]{rice-org}
Henry~Grodon Rice.
\newblock Classes of recursively enumerable sets and their decision problems.
\newblock {\em Transactions of the American Mathematical Society},
  74(2):358--366, 1953.

\bibitem[vN66]{neumann}
John von Neumann.
\newblock {\em Theory of Self-Reproducing Automata}.
\newblock University of Illinois Press, Champaign, IL, USA, 1966.

\end{thebibliography}

\end{document}